\documentclass[11pt]{article}
\usepackage{fullpage}

\usepackage{algorithm, algorithmic}
\usepackage{comment,amsfonts,amsmath,amsthm,graphicx, algorithm, algorithmic}
\usepackage{boldline}
\usepackage{authblk}

\newcommand{\commentout}[1]{}

\def\tO{\tilde{O}}
\def\epsilon{\varepsilon}

\ifx\pdftexversion\undefined
\usepackage[colorlinks,linkcolor=black,filecolor=black,citecolor=black,urlco
lor=black,pdfstartview=FitH]{hyperref}
\else
\usepackage[colorlinks,linkcolor=blue,filecolor=blue,citecolor=blue,urlcolor
=blue,pdfstartview=FitH]{hyperref}
\fi

\newcommand{\alert}[1]{\textbf{\color{red}
		[[[#1]]]}\marginpar{\textbf{\color{red}**}}\typeout{ALERT:
		\the\inputlineno: #1}}

\def\tO{\tilde{O}}

\newcommand{\PP}{\mathcal{P}}
\newcommand{\R}{\mathbb{R}}

\newcommand{\E}{{\mathbb{E}}}

\newcommand{\mommit}[1]{}
\newcommand{\namedref}[2]{\hyperref[#2]{#1~\ref*{#2}}}
\newcommand{\sectionref}[1]{\namedref{Section}{#1}}

\newcommand{\figureref}[1]{\namedref{Figure}{#1}}
\newcommand{\algref}[1]{\namedref{Algorithm}{#1}}

\newcommand{\lemmaref}[1]{\namedref{Lemma}{#1}}

\newtheorem{theorem}{Theorem}
\newtheorem{lemma}{Lemma}

\title{Improved Weighted Additive Spanners}

\author[1]{Michael Elkin}
\author[1]{Yuval Gitlitz}
\author[1]{Ofer Neiman}

\affil[1]{Department of Computer Science, Ben-Gurion University of the Negev,
	Beer-Sheva, Israel. Email: \texttt{\{elkinm,neimano\}@cs.bgu.ac.il}, \texttt{gitlitz@post.bgu.ac.il }}

\begin{document}

\maketitle{}

\begin{abstract}
Graph spanners and emulators are sparse structures that approximately preserve distances of the original graph. While there has been an extensive amount of work on additive spanners, so far little attention was given to weighted graphs. Only very recently \cite{ABSKS20} extended the classical +2 (respectively, +4) spanners for unweighted graphs of size $O(n^{3/2})$ (resp., $O(n^{7/5})$) to the weighted setting, where the additive error is $+2W$ (resp., $+4W$). This means that for every pair $u,v$, the additive stretch is at most $+2W_{u,v}$, where $W_{u,v}$ is the maximal edge weight on the shortest $u-v$ path (weights are normalized so that the minimum edge weight is 1).  In addition, \cite{ABSKS20} showed a randomized algorithm yielding a $+8W_{max}$ spanner of size $O(n^{4/3})$, here $W_{max}$ is the maximum edge weight in the entire graph.

In this work we improve the latter result by devising a simple deterministic algorithm for a $+(6+\varepsilon)W$ spanner for weighted graphs with size $O(n^{4/3})$ (for any constant $\varepsilon>0$), thus nearly matching the classical +6 spanner of size $O(n^{4/3})$ for unweighted graphs. Furthermore, we show a $+(2+\varepsilon)W$ subsetwise spanner of size $O(n\cdot\sqrt{|S|})$, improving the $+4W_{max}$ result of \cite{ABSKS20} (that had the same size). We also show a simple randomized algorithm for a $+4W$ emulator of size $\tO(n^{4/3})$.

In addition, we show that our technique is applicable for very sparse additive spanners, that have linear size. It is known that such spanners must suffer polynomially large stretch. 
For weighted graphs, we use a variant of our simple deterministic algorithm that yields a linear size  $+\tO(\sqrt{n}\cdot W)$ spanner, and we also obtain a tradeoff between size and stretch.

Finally, generalizing the technique of \cite{DHZ00} for unweighted graphs, we devise an efficient randomized algorithm producing a $+2W$ spanner for weighted graphs of size $\tO(n^{3/2})$ in $\tO(n^2)$ time. 
\end{abstract}

\section{Introduction}
Let $G= (V, E, w)$ be a weighted undirected graph on $n$ vertices. Denote by $d_G(u, v)$ the distance between $u, v \in V$ in the graph  $G$. A graph $H = (V, E', w)$ is an {\em $(\alpha, \beta)$-spanner} of $G$ if it is a subgraph of $G$ and for every $u, v \in V$,
\[
	d_H(u, v) \le \alpha \cdot d_G(u, v) + \beta.
\]
For an {\em emulator} $H$, we drop the subgraph requirement (that is, we allow $H$ to have edges that are not present in $G$, while still maintaining $d_H(u,v)\ge d_G(u,v)$ for all $u,v\in V$).

Spanners were introduced in the 80's by \cite{PS89}, and have been extensively studied ever since. One of the key objectives in this field is to understand the tradeoff between the stretch of a spanner and its size (number of edges). For purely multiplicative spanners (with $\beta=0$), an answer was quickly given: for any integer $k\ge 1$, \cite{DBLP:journals/dcg/AlthoferDDJS93} showed that a greedy algorithm provides a $(2k-1, 0)$-spanner with size $O(n^{1 + 1/k})$. This bound is tight assuming Erd\H{o}s' girth conjecture.

In this paper we focus on purely additive spanners, where $\alpha=1$, which we denote by $+\beta$ spanners.
Almost all of the previous work on purely additive spanners was done for unweighted graphs. The first purely additive spanner was a $+2$ spanner of size $O(n^{1.5})$ \cite{DBLP:journals/siamcomp/AingworthCIM99, EP04}, which was followed by a $+6$ spanner of size $O(n ^ {4/3})$ \cite{DBLP:conf/soda/BaswanaKMP05, DBLP:conf/swat/Knudsen14}, and a $+4$ spanner of size $O(n^{7/5})$ \cite{DBLP:conf/soda/Chechik13,DBLP:journals/corr/abs-2001-07741}. A result of \cite{DBLP:journals/jacm/AbboudB17} showed that any purely additive spanner with $O(n^{4/3-\delta})$ edges, for constant $\delta>0$, must have a polynomial stretch $\beta$. On the other hand, several works \cite{P09,DBLP:conf/soda/Chechik13,BW15,BW16} obtained sparser spanners with polynomial stretch. The state-of-the-art result of \cite{BW16} has near-linear size and stretch $\tO(n^{3/7})$.

In \cite{EP04} the notion of {\em near-additive} spanners for unweighted graphs was introduced, where $\alpha=1+\epsilon$ for some small $\epsilon>0$. They showed $(1+\epsilon,\beta)$-spanners of size $O(\beta\cdot n^{1+1/k})$ with $\beta = O(\frac{\log k}{\epsilon})^{\log k}$. Many following works \cite{E01,EZ04,TZ06,P09,ABP17,EN16} improved several aspects of these spanners, but up to the $\beta$ factor in the size, this is still the state-of-the-art. Providing some evidence to its tightness, \cite{ABP17} showed that such spanners must have $\beta = \Omega(\frac{1}{\epsilon\cdot \log k})^{\log k}$.

Since many applications of spanners stem from weighted graphs (see \cite{ABSKS20} and the references therein), it is only natural to study additive spanners in that setting. Assume the weights are normalized so that the minimum edge weight is 1. We distinguish between two types of additive spanners; in the first one the additive stretch is $+c \cdot W_{\max}$, where $W_{\max}$ is the weight of heaviest edge in the graph, and $c$ is usually some constant.
A more desirable type of additive stretch is denoted by $+c\cdot W$, which means that for every $u, v \in V$,
\[
	d_H(u, v) \le d_G(u, v) + c \cdot W_{u, v},
\]
where $W_{u, v}$ is the heaviest edge in the shortest path between $u, v$ in $G$.
This estimation is not only stronger, but also handles nicely the multiplicative perspective of the spanner: a $+c \cdot W$ spanner is also a $(c + 1, 0)$ spanner (while a $+W_{max}$ approximation can have unbounded multiplicative stretch).

The first adaptation of (near)-additive spanners to the weighted setting was given in \cite{DBLP:journals/corr/abs-1907-11422}, where we showed near-additive spanners and emulators with essentially the same stretch and size as the state-of-the-art results for unweighted graphs, while $\beta$ is multiplied by $W$ (the maximal edge weight on the corresponding path).
In addition, a construction of an additive $+2W$ spanner of size $\tilde{O}(n^{3/2})$ can be inferred from \cite{DBLP:journals/corr/abs-1907-11422}.\footnote{The notation $\tilde{O}(\cdot)$ hides polylogarithmic factors.}
Ahmed et al.  \cite{ABSKS20} recently gave a comprehensive study of weighted additive spanners. Among other results, they showed a $+2W_{\max}$ spanner of size $O(n^{1.5})$, a $+4W$ spanner of size $O(n^{7/5})$,
\footnote{In their paper the spanner is claimed to be $+4W_{\max}$ but a tighter analysis shows it is actually a $+4W$.} and a $+8 W_{\max}$ spanner of size $O(n^{4/3})$. Given a set $S\subseteq V$, they showed a $+4W_{max}$ subsetwise spanner of size $O(n\cdot\sqrt{|S|})$ (that has approximation guarantee only for pairs in $S\times S$). While the former two results match the state-of-the-art unweighted bounds, the latter two leave room for improvement. Indeed, \cite{ABSKS20} pose as an open problem whether a $+6W_{\max}$ spanner of size $O(n^{4/3})$ can be achieved.



\paragraph{Our results.} In this work we improve the bounds of \cite{ABSKS20} both quantitatively and qualitatively. For any constant $\varepsilon > 0$, we show a simple deterministic construction of a $+(6 + \varepsilon)W$ spanner of size $O(n^{4/3})$.\footnote{For arbitrary $\varepsilon > 0$, the size of our spanner is $O(n^{4/3}/\varepsilon)$.} Thus, the additive stretch of our spanner is arbitrarily close to $6W$, while having the superior dependence on the largest edge weight on the shortest $u-v$ path, rather than the global maximum weight. Furthermore, our algorithm is a simple greedy algorithm, in contrast to the more involved 2-stages randomized algorithm of \cite{ABSKS20}.

We show the versatility of our techniques by applying them to the subsetwise setting. Given a set $S\subseteq V$, for any constant $\varepsilon>0$, we obtain a $(2+\varepsilon)\cdot W$ subsetwise spanner of size $O(n\cdot \sqrt{|S|})$, again improving \cite{ABSKS20} both in the stretch and in the dependence on maximal edge weight.

A slight variant of our simple greedy algorithm works in the setting of sparse spanners with polynomial additive stretch, also for weighted graphs. This is in contrast to essentially all previous algorithms for very sparse pure additive spanners, that were rather involved. In particular, we obtain a linear size $+\tO(\sqrt{n})\cdot W$ spanner, and more generally, for any $0\le\varepsilon\le 1$, a $+O(n^{\frac{1-\varepsilon}{2}} \log n)W$ spanner of size $O(n^{1 + \varepsilon})$.
While this result does not match the state-of-the-art for unweighted graphs, we believe it is interesting to have such spanners in the weighted setting, and we find the simplicity of the algorithm appealing.

In addition, we show a simple randomized algorithm that produces a $+4W$ emulator of size $\tilde{O}(n^{4/3})$. This corresponds to the $+4$ emulator of size $O(n^{4/3})$ for unweighted graphs \cite{DBLP:journals/siamcomp/AingworthCIM99,EP04}.

Finally, bearing the mind the applications of such spanners to efficiently computing shortest paths, we devise an efficient $\tO(n^2)$ time algorithm for a $+(2+\varepsilon)W$ spanner of size $\tO(n^{3/2})$. This result builds on the \cite{DHZ00} $+2$ spanner for unweighted graphs.


\paragraph{Overview of our construction and analysis.}

Our algorithms for the $(6+\varepsilon)\cdot W$ spanner and the $(2+\varepsilon)\cdot W$ subsetwise spanner follow a common approach.
We adapt the algorithm of \cite{DBLP:conf/swat/Knudsen14}, who showed a simple $+6$ spanner for unweighted graphs, to the weighted setting. Both \cite{DBLP:conf/swat/Knudsen14} and the path-buying construction of \cite{DBLP:conf/soda/BaswanaKMP05} iteratively add paths to the spanner $H$, and argue that for each new edge in a path that is added to $H$, there is some progress for many pairs of vertices. Specifically,
assume that for some $u, v \in V$ we have for a constant $c$ that
\begin{equation}\label{eq:1}
	d_H(u, v) \le d_G(u, v) + c~,
\end{equation}
where $H$ is the current spanner we maintain.
For unweighted graphs, if we make progress and improve the distance in $H$ between $u,v$, it will be by at least 1. Thus, once we obtain \eqref{eq:1}, the distance between $u, v$ can be improved at most $c$ more times. This nice attribute does not apply to weighted graphs, since there the distance between $u, v$ can be improved only by a tiny amount.

In our algorithm, we first add the $t$-lightest edges incident on every vertex (the value of $t$ depends on the required sparsity), and then greedily add shortest paths between vertices whose stretch is too large, ordered by their $W$. To overcome the issue of tiny improvements, our notion of progress depends on the weights. That is, when adding paths to the spanner, we will show that many pairs improve their distance by at least $\Omega(\varepsilon \cdot W)$. Note that $W$ is in fact a function (the maximum edge weight in the current path), so some care is required to ensure sufficient progress is made for many other pairs (that can have either a smaller or a larger $W$).
Now, if the current distance in $H$ between $u, v \in V$ is
\[
	d_H(u, v) \le d_G(u, v) + c \cdot W,
\]

then the distance between $u, v$ can be improved at most $O(\frac{c}{\varepsilon})$ more times. This number translates directly to the size of the spanner, and also affects the stretch. 

While our linear size $+\tO(\sqrt{n})\cdot W$ spanner is also built using a similar greedy algorithm (augmented by a multiplicative spanner), its analysis is more involved. We use a labeling scheme of the graph vertices. The idea is that each of the greedily added paths must have labeled a lot of new vertices, else we could have used the existing $t$-lightest edges, combined with the multiplicative spanner and the previously added paths, to obtain a sufficiently low stretch alternative path. We then conclude that the number of added paths is bounded, which is then used to bound the number of edges added to the spanner in all these paths, by an argument based on low intersections between shortest paths.


\subsection{Organization}

After reviewing a few preliminary results in \sectionref{sec:prel}, we show our $+(6+\varepsilon)\cdot W$ spanner in \sectionref{sec:6span}, and the linear size spanner with polynomial stretch for weighted graphs in \sectionref{sec:poly}. The $+2W$ spanner with $\tilde{O}(n^2)$ construction time is shown in \sectionref{sec:fast}. Our $+(2+\varepsilon)\cdot W$ subsetwise spanner is in \sectionref{sec:subset}, and the $+4W$ emulator in \sectionref{sec:4emul}.

\section{Preliminaries}\label{sec:prel}

Let $G=(V,E,w)$ be a weighted undirected graph, with nonnegative weights $w:E\to\R_+$ 
, and fix a parameter $\varepsilon > 0$.
Denote by $P_{u, v}$ the shortest path between vertices $u, v\in V$, breaking ties consistently (say by id's), so that every sub-path of a shortest path is also a shortest path and two shortest paths have at most one intersecting subpath.
Let $W_{u, v}$ denote the weight of the heaviest edge in $P_{u, v}$. For a positive integer $t$, a {\em $t$-light initialization} of $G$ is a subgraph $H=(V,E',w)$ that contains, for each $u\in V$, the lightest $t$ edges incident on $u$ (or all of them, if $\deg(u)\le t$), breaking ties arbitrarily. For $u\in V$, we say that $v$ is a {\em $t$-light neighbor} of $u$ if the edge $\{u,v\}$ is contained in a  $t$-light initialization of $G$.

The following lemma was shown in \cite[Theorem 5]{ABSKS20}.
\begin{lemma} [\cite{ABSKS20}]\label{lem:a_lot_of_neighbors}
Let $G = (V, E, w)$ be an undirected weighted graph, and $H$ a $t$-light initialization of $G$. If
$P_{u,v}$ is some shortest path in $G$ that is missing $\ell$ edges in $H$, then there is a set of vertices $S\subseteq V$ such that:
\begin{enumerate}
	\item $|S| = \Omega(t\ell)$.
	\item Each vertex of $S$ has a $t$-light neighbor in $P_{u, v}$, with edge weight at most $W_{u,v}$.
\end{enumerate}
\end{lemma}
(The fact that {\em light} edges are connecting $S$ to $P_{u, v}$ did not appear explicitly in \cite{ABSKS20}, but it follows directly from their proof.)

We will also use the construction of the greedy spanner multiplicative spanners \cite{DBLP:journals/dcg/AlthoferDDJS93}.

\begin{lemma} (\cite{DBLP:journals/dcg/AlthoferDDJS93}) \label{lem:multspan}
Let $G = (V, E, w)$ be an undirected weighted graph, and fix a parameter $k \ge 1$. There exists a $(2k-1, 0)$-spanner of size $O(n ^{1 + 1 / k})$.
\end{lemma}

The following standard lemma asserts that sampling a random set $S$ of vertices with the appropriate density, will guarantee with high probability (w.h.p.) that for every $u \in V$: either all of its neighbors are in a $t$-light initialization, or $u$ has a light neighbor in $S$.

\begin{lemma}\label{lem:emu:dominating}
	Let $G = (V, E, w)$ be an undirected weighted graph and let $H$ be a $(2 n^{\varepsilon} \ln n)$-light initialization of $G$ for some $0\le\varepsilon\le 1$.
	Let $S\subseteq V$ be a random set, created by sampling each vertex independently with probability $\frac{1}{n^\varepsilon}$.
	Then with probability at least $1-1/n$, for every vertex $u$ having at least $2n^{\varepsilon} \ln n$ neighbors in $G$, there exists $y \in S$ s.t. $y$ is a $(2n^{\varepsilon} \ln n)$-light neighbor of $u$.
	
\end{lemma}
\begin{proof}
	Let $U$ be the set of vertices with degree at least $2n^{\varepsilon} \ln n$ in $G$.
	Fix $u \in U$, and denote by $X_u$ the event that there exists $y \in S$ which is a $(2n^{\varepsilon} \ln n)$-light neighbor of $u$.
	Every vertex is sampled to $S$ independently with probability $\frac{1}{n^{\varepsilon}}$, hence
	\[
	\Pr[\bar{X_u}] = \left(1- \frac{1}{n^{\varepsilon}}\right)^{2n^{\varepsilon}\ln n} \le (1/e)^{2 \ln n} = (1/n)^2.
	\]
	
	Let $X$ be the event that for every $u \in U$, the event $X_u$ occur. By the union bound,
	\[
	\Pr[\bar{X}] \le \sum_{u \in U} \Pr[\bar{X_u}] \le |U| /n^2 \le 1/n.
	\]
	

\end{proof}

\section{A $+(6 + \varepsilon)W$ spanner}\label{sec:6span}

\paragraph{Construction.}
Our algorithm for a $+(6 + \varepsilon)W$ spanner works as follows. Initially, $H$ is set as a $n^{1/3}$-light initialization of $G$. Next, sort all the pairs $u,v\in V$: first according to $W_{u,v}$, and then by $d_G(u,v)$ (from small to large), breaking ties arbitrarily.
Then, go over all pairs in this order; when considering $u,v$, we add $P_{u, v}$ to $H$ if
\begin{align}
	\label{eq:when_puv_added}
d_H(u, v) > d_G(u, v) + (6 + \varepsilon)W_{u, v}.
\end{align}

\paragraph{Analysis.} Our main technical lemma below asserts that by adding a shortest path to $H$, we get for many pairs of the path's neighbors: 1) a good initial guarantee, and also 2) sufficiently improve their distance in $H$.

\begin{figure}[h]
	\centering
	\includegraphics{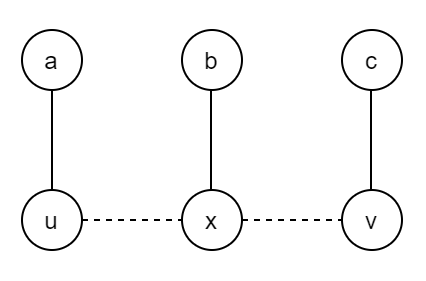}	
	\caption{An illustration for \lemmaref{lem:improved_dist}. The dotted line is $P_{u, v}$, and the edges $\{a, u\}, \{b, x\}, \{c, v\}$ are all light. It is possible that $u = x$ or $v = x$.}
	\label{fig:lem:improved_dist}
\end{figure}

\begin{lemma}\label{lem:improved_dist}
	Let $u, v \in V$ be two vertices for which the path $P_{u, v}$ was added to $H$, and take any $x\in P_{u,v}$. Let $a, b, c \in V$ be different $n^{1/3}$-light neighbors of $u, x, v$, respectively, with edge weights at most $W_{u,v}$.
	Denote by $H_0$ the spanner just before $P_{u, v}$ was added and by $H_1$ the spanner right after the path was added. Then both of the following holds.
	\begin{enumerate}
		
		\item $
		d_{H_1}(a, b) \le d_G(a, b) + 4 W_{u, v}
		\textbf{ and }
		d_{H_1}(b,c) \le d_G(b,c) + 4 W_{u, v}$.
		\item
		$
		d_{H_1}(a, b) \le d_{H_0}(a, b) - \frac{\varepsilon}{2} W_{u, v}
		\textbf{ or }
		d_{H_1}(b,c) \le d_{H_0}(b,c) - \frac{\varepsilon}{2} W_{u, v}
		$.

		\label{lem:improved_distance}
	\end{enumerate}
\end{lemma}
\begin{proof}
Fix $P_{u,v}$ and $a,b,c$ as defined in the Lemma, see also \figureref{fig:lem:improved_dist}.
	We begin by proving the first item, using the triangle inequality and the fact that the three edges $\{a,u\},\{b,x\},\{c,v\}$ all appear in $H_1$ (since they are $n^{1/3}$-light), and have weight at most $W_{u,v}$.
	\begin{eqnarray}
		d_{H_1}(a, b)\nonumber
		&\le& d_{H_1}(a, u) + d_{H_1}(u, x) + d_{H_1}(x, b) \\\label{eq:a,b}
		&=& w(a, u) + d_{G}(u, x) + w(x, b) \\\nonumber
		&\le& w(a, u) + d_{G}(u, a) +d_G(a, b) + d_G(b, x) + w(x, b)\\\nonumber
		&\le& d_G(a, b) + 4 W_{u, v}.
	\end{eqnarray}	
	The bound on $d_{H_1}(b,c)$ follows in a symmetric manner, which concludes the proof of the first item.
	Seeking contradiction, assume that the second item does not hold. This suggests that
	\[
	d_{H_0}(a, b)< d_{H_1}(a, b)+ \frac{\varepsilon}{2} W_{u, v}\stackrel{\eqref{eq:a,b}}{\le} d_G(u, x) + (2 + \frac{\varepsilon}{2})W_{u, v}~,
	\]

	and also
	\[
	d_{H_0}(b,c) <  d_{H_1}(b,c)+ \frac{\varepsilon}{2} W_{u, v}\le d_G(x, v) + (2 + \frac{\varepsilon}{2})W_{u, v}~.
	\]
	 So we have that
	\begin{align*}
		d_{H_0}(u, v)
		&\le d_{H_0}(u, a) + d_{H_0}(a, b) + d_{H_0}(b, c) + d_{H_0}(c, v) \\
		&< w(u, a) + d_{G}(u, x) + (2 + \frac{\varepsilon}{2})W_{u, v}  + d_{G}(x, v) +(2 + \frac{\varepsilon}{2})W_{u, v} + w(c, v) \\
		&\le d_G(u, v) + (6 + \varepsilon)W_{u, v},
	\end{align*}
	which is a contradiction to \eqref{eq:when_puv_added}, since we assumed that the path $P_{u, v}$ was added to the spanner.


\end{proof}

\begin{theorem}
	\label{thm:6eps_span}
	For every undirected weighted graph $G= (V, E, w)$ and $\varepsilon>0$, there exists a deterministic polynomial time algorithm that produces a $+(6 + \varepsilon)W$ spanner of size $O(\frac{1}{\varepsilon}\cdot n^{4/3})$.
\end{theorem}
\begin{proof}

Our construction algorithm adds a shortest path between pairs whose stretch is larger than $+(6 + \varepsilon)W$, so we trivially get a $+(6 + \varepsilon)W$ spanner (the running time can be easily checked to be polynomial in $n$). Thus, we only need to bound the number of edges. Starting with the $n^{1/3}$-light initialization introduces at most $n^{4/3}$ edges to the spanner, so it remains to bound the number of edges added by adding the shortest paths.

Let $u, v \in V$ be two vertices for which the path $P_{u, v}$ was added to the spanner. Consider the time in which this path was added, let $H_0$ be the spanner just before the addition of $P_{u,v}$, and $H_1$ after the addition. We say that a pair of vertices $a,b\in V$ is {\em set-off} at this time, if it is the first time that $d_{H_1}(a, b) \le d_G(a, b) + 4 W_{u, v}$, and it is {\em improved} if $d_{H_1}(a, b) \le d_{H_0}(a, b) - \frac{\varepsilon}{2} W_{u, v}$. The main observation is that once a pair is set-off, it can be improved at most $O(\frac{1}{\varepsilon})$ times. To see this, note that after the set-off we have $d_H(a,b)-d_G(a,b)\le 4W_{u,v}$, and recall that we ordered the pairs by their maximal weight $W_{u,v}$, so any future improvement will be at least by $\frac{\varepsilon}{2} W_{u, v}$. Since at the end we must have $d_H(a,b)\ge d_G(a,b)$, there can be at most $O(\frac{1}{\varepsilon})$ improvements.

We will show that if $\ell$ edges of $P_{u, v}$ are missing in $H_0$, then at least $\Omega(\ell\cdot n^{2/3})$ pairs are either set-off or improved. Fix any $x\in P_{u,v}$, and
let $a, b, c \in V$ be different $n^{1/3}$-neighbors of $u, x, v$, respectively, connected by edges of weight at most $W_{u,v}$. Apply \lemmaref{lem:improved_distance} on $u,v,x$ and $a, b, c$. 
We get that both pairs $(a,b)$ and $(b,c)$ are set-off (if they haven't before), and at least one of them is improved.

The final goal is to show that there are $\Omega(\ell\cdot n^{2/3})$ such set-off/improving pairs. We first claim that the first and last edges of $P_{u,v}$ are missing in $H_0$. Seeking contradiction, assume that the first edge $\{u,u_1\}\in E(H_0)$, then the pair $u_1,v$ has $W_{u_1,v}\le W_{u,v}$ and $d_G(u_1,v)<d_G(u,v)$ (using that the sub-path of $P_{u,v}$ from $u_1$ to $v$ is the shortest path between $u_1,v$), and its stretch must be larger than $+(6 + \varepsilon)W_{u,v}$ (otherwise $u,v$ will have stretch at most $+(6 + \varepsilon)W_{u,v}$ as well), so we should have considered the pair $u_1,v$ before $u,v$, and added $P_{u_1,v}$ to $H$. That would produce a shortest path between $u,v$, which yields a contradiction to \eqref{eq:when_puv_added}. A symmetric argument shows that the last edge is missing too.

Now, since $H_0$ contains a $n^{1/3}$-light initialization, but $u$ (resp., $v$) has a missing edge, it follows that $u$ (resp., $v$) has at least $n^{1/3}$ neighbors that are all lighter than the missing first (resp., last) edge of $P_{u,v}$, and thus of weight at most $W_{u,v}$. So there are at least $n^{1/3}$ choices for $a$ and for $c$. By \lemmaref{lem:a_lot_of_neighbors} there are at least $\Omega(\ell\cdot n^{1/3})$ choices for $b$. We conclude that there are at least $\Omega(\ell\cdot n^{1/3} \cdot n^{1/3}) = \Omega(\ell\cdot n^{2/3})$ pairs that are set-off/improved.

Let $t$ be the number of edges added by all paths. Since every pair can be set-off only once, and improved $O(\frac{1}{\varepsilon} )$ times, we get the following inequality

\[
	\Omega(t\cdot n^{2/3}) \le O(\frac{n^2}{\varepsilon} )~,
\]

thus $t = O(\frac{ n^{4/3}}{\varepsilon})$.
\end{proof}


\section{A $+\tO(n^{\frac{1 - \varepsilon}{2}}W)$ spanner of size $O(n^{1 + \varepsilon})$}\label{sec:poly}

Let $G=(V,E,w)$ be a weighted undirected graph with $n$ vertices, and let $0\le\varepsilon\le 1$ be a parameter. We will now present our $+O(n^{\frac{1-\varepsilon}{2}} \log n)$ spanner of size $O(n^{1 + \varepsilon})$.

\paragraph*{Construction.}
Let $H$ be a $(n^\varepsilon)$-light initialization of $G$.
We then add the edges of $(\log n, 0)$-greedy spanner from Lemma \ref{lem:multspan} to $H$. Next, we sort all the pairs $u, v \in V$ by $W_{u, v}$ in increasing order (breaking ties arbitrarily). For each pair $(u, v)$ we add $P_{u, v}$ if 
\begin{align}
\label{eq:poly_span}
	d_H(u, v) > d_G(u, v) + c \cdot n^{\frac{1- \varepsilon}{2}} \log n\cdot W_{u,v},
\end{align}
where $c$ is a constant to be determined.

\paragraph*{Analysis.} By the last step of the algorithm, every pair will have stretch $O(n^{\frac{1- \varepsilon}{2}} \log n\cdot W)$. The number of edges added by the $(n^\varepsilon)$-light initialization of $G$ is at most $n^{1+\varepsilon}$, and the $(\log n, 0)$-greedy spanner from Lemma \ref{lem:multspan} has $O(n)$ edges. The main difficulty of the analysis lies in bounding the number of edges in the paths added by the algorithm. Denote by $\PP$ the set of paths added in the last stage. We start by bounding the number of such paths.

\begin{lemma}\label{lem:small_num_paths}
	$|\PP| \le n^{\frac{1-\varepsilon}{2}}.$
\end{lemma}
\begin{proof}

We will define a labeling for the vertices. At the beginning, all the vertices will be unlabeled.
Go over the added paths by the order of the algorithm.
For every path $P_{x,y}$ which was added to the spanner, and every missing edge $(a,b)$ in it, we label by $\{x,y\}$ all the unlabeled $(n^\varepsilon)$-light neighbors of $a$ and of $b$.
We will show that for every added path, we label at least $n^{\frac{1+\varepsilon}{2}}$ vertices.
This will imply that
\[
	|\PP| \le \frac{n}{n^{\frac{1+\varepsilon}{2}}} =  n^{\frac{1-\varepsilon}{2}},
\]
proving the lemma.

Seeking contradiction, assume that there is a path for which we labeled less than $n^{\frac{1 + \varepsilon}{2}}$ vertices, and let $P_{u,v}$ be the first such path considered by the algorithm. Note that there can be at most $n^{\frac{1-\varepsilon}{2}}$ paths that were added before $P_{u,v}$.

Let $H_0$ be the spanner just before $P_{u, v}$ was added. The goal is to show a low stretch path in $H_0$ between $u,v$, contradicting the fact that $P_{u,v}$ was added. To this end, we distinguish between two types of edges in $P_{u,v}$ that are missing in $H_0$. 

The first type are missing edges $(a,b)$ that all the $(n^\varepsilon)$-light neighbors of $a$ or all the $(n^\varepsilon)$-light neighbors of $b$ are unlabeled.
Observe that there is a constant $k$, so there can be at most $k \cdot n^{\frac{1 - \varepsilon}{2}}$ such missing edges, since by \lemmaref{lem:a_lot_of_neighbors} $k \cdot n^{\frac{1-\varepsilon}{2}}$ missing edges have at least $\Omega(k \cdot n^{\frac{1-\varepsilon}{2}} \cdot n^\epsilon)  = \Omega(k \cdot n^{\frac{1 - \varepsilon}{2}})$ neighbors which are given labels. Choosing a large enough $k$, will contradict the assumption we label less than $n^{\frac{1 + \varepsilon}{2}}$ vertices when adding $P_{u,v}$.
So for each such an edge $(a,b)$ we can use the $\log n$-spanner which gives stretch at most $\log n\cdot w(a,b)\le\log n\cdot W_{u,v}$. Thus the total stretch over all these edges is at most $k \log n\cdot n^{\frac{1 - \varepsilon}{2}}\cdot W_{u,v}$.

The second type are missing edges with a labeled $(n^\varepsilon)$-light neighbor. Suppose $u'$ is a vertex in $P_{u,v}$ on a missing edge $(u',u'')$ with an $(n^\varepsilon)$-light neighbor labeled $\{x,y\}$. Let  $v'$ be the rightmost vertex on a missing edge $(v'',v')$ in $P_{u,v}$ with an $(n^\varepsilon)$-light neighbor labeled by $\{x,y\}$. Denote by $a$ (resp. $b$) the light neighbor of $u'$ (resp. $v'$) with label $\{x,y\}$.  Let $x'$ (resp., $y'$) be a vertex in $P_{x,y}$ such that $a$ (resp., $b$) is a $(n^\varepsilon)$-light neighbor of $x'$ (resp., $y'$) (see \figureref{fig:lem:small_num_paths}). Note that $w(u',a)\le w(u',u'')\le W_{u,v}$, since the edge $(u',u'')$ was not added in the $(n^\varepsilon)$-initialization, and similarly $w(v',b)\le W_{u,v}$. Also $w(x',a)\le W_{x,y}\le W_{u,v}$, since $a$ got its label by being a light neighbor of a missing edge in $P_{xy}$, and $W_{x,y}\le W_{u,v}$ by the initial sort of pairs according to the heaviest edge. Similarly $w(y',b)\le W_{u,v}$. Recalling that all the edges to an $(n^\varepsilon)$-light neighbor are in $H_0$, we can now see that the distance between $u'$ and $v'$ in $H_0$ has constant additive stretch:
\begin{align*}
	d_{H_0}(u', v') &\le d_{H_0}(u', a) + d_{H_0}(a, x') + d_{H_0}(x', y') + d_{H_0}(y', b) + d_{H_0}(b, v') \\
	&\le d_G(u', a) + d_G(a, x') + d_G(x', y') + d_G(y', b) + d_G(b, v') \\
	&\le 2(d_G(u', a) + d_G(a, x')) + d_G(u', v') + 2(d_G(y', b) + d_G(b, v')) \\
	&\le d_G(u', v') + 8W_{u, v}.
\end{align*}

We conclude that whenever we encounter a vertex $u'$ on a missing edge with a light neighbor labeled $\{x,y\}$, we can simply use the path in $H_0$ to the last vertex $v'$ on $P_{u,v}$ on a missing edge with a light neighbor labeled $\{x,y\}$, and pay only $8W_{u,v}$ additive stretch. Let $z$ be the neighbor of $v'$ closer to $v$, then use the multiplicative spanner in case the edge $(v',z)$ is missing. The remaining path from $z$ to $v$ will clearly have no more missing edges with a light neighbor labeled $\{x,y\}$. Recall that we added at most $ n^{\frac{1 - \varepsilon}{2}}$ paths before $P_{u,v}$, so there can be at most  $n^{\frac{1 - \varepsilon}{2}}$ different labels. This suggests that the total additive stretch accumulated by the second type of missing edges is at most $(8+
\log n)\cdot n^{\frac{1 - \varepsilon}{2}}\cdot W_{u,v}$.

Thus there exists a path in $H_0$ between $u,v$ of length at most $d_G(u,v)+(8+(1 + k)\log n)\cdot n^{\frac{1 - \varepsilon}{2}}\cdot W_{u,v}$, setting $c\ge 9 + k$ contradicts the fact that $P_{u,v}$ was added by the algorithm. This concludes the proof of the lemma.

\begin{figure}[h]
	\centering
	\includegraphics[scale=0.7]{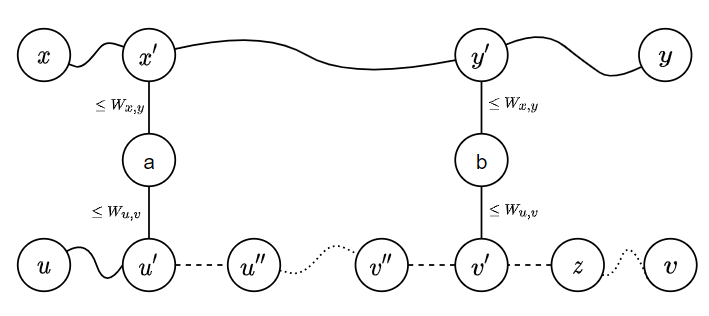}
	\caption{An illustration for \lemmaref{lem:small_num_paths}. Straight lines and curved lines are edges and paths which are present in $H_0$. Dotted straight lines are edges missing in $H_0$ and dotted curved lines are path with possibly missing edges in $H_0$}
	\label{fig:lem:small_num_paths}
\end{figure}

\end{proof}

\begin{lemma} \label{lem:paths_low_cost}
	Adding $\PP$ to $H$ adds $O(n)$ edges to the spanner.	
\end{lemma}
\begin{proof}
	Let $P_{u, v}$ be a path added by the algorithm. Let $H_0$ be the spanner just before it is added. Then for every edge $(a, b) \in P_{u, v}$ there are three cases:
	\begin{enumerate}
		\item At least one of the vertices $a, b$ does not belong to any path previously added to $H$. Since every vertex has 2 edges touching it in the path, there can be at most $2n$ such edges.
		\item Both $a,b$ belong to the same previously added path. Note that the edge $(a, b)$ is already in $H_0$ in this case.
		\item There is a previously added path $P_{x,y}$ such that $a\in P_{x,y}$ and $b\notin P_{x,y}$. Then the two paths $P_{x,y}$ and $P_{u,v}$ start their intersection at $a$.
	\end{enumerate}
To bound the number of edges in case 3, note that every two paths can have only one intersecting subpath. So any pair of paths in $\PP$ can introduce at most 2 edges to case 3 (the first and the last edge in their common subpath). By \lemmaref{lem:small_num_paths} there can be at most $2 {|\PP| \choose 2} = O(n^{1 - \varepsilon})$ such added edges in all the paths.
\end{proof}
By \lemmaref{lem:paths_low_cost} the number of edges in $H$ is $O(n^{1 + \varepsilon})$. We have proven the following theorem.
\begin{theorem}
	For every undirected weighted graph $G = (V, E, w)$ and $0\le\varepsilon \le 1$, there exists a deterministic polynomial time algorithm that produces a $+O(n^{\frac{1-\varepsilon}{2}} \log n)W$ spanner of size $O(n^{1 + \varepsilon})$.
\end{theorem}

\section{A +2W spanner in $\tO(n^2)$ time}\label{sec:fast}

Let $G=(V,E,w)$ be a weighted graph with $n$ vertices, and fix $k = 1/2 \cdot \log n$ (assume $k$ is an integer).
Set $s_0 = n, s_1 = n/2,\ldots,s_k = n/2^k=\sqrt{n}$.
For each $i = 0,1,\ldots,k$, let $V_i$ be the set of vertices  of degree at least $s_i$ (note that $V_0 = \emptyset$), set $V_{k+1}=V$.
Let $D_i$ be a set of vertices sampled independently at random from $V$, each with probability $p={{c\log n} \over {s_i}}$ for a constant $c>1$. By standard considerations it follows that
w.h.p. $|D_i| = \Theta({n\log n \over {s_i}})$, and $D_i$ is a dominating set for $V_i$ by \lemmaref{lem:emu:dominating}.

For every $i \in [k]$, and for every $v \in V_i$, let $p_i(v) \in D_i$ be the closest vertex in $D_i$ to $v$ (breaking ties arbitrarily).
Define $E^*_i = \{(v,p_i(v)) ~:~ v \in V_i\}$.
Also, for every $v \in V_i$, 
define $Bunch_i(v)  = \{(u,v)\in E ~:~ w((u,v))<w((v,p_i(v)))\}$. 
For $v \not \in V_i$, (i.e., $deg(v) < s_i$), set $Bunch_i(v) = \{(v,u) \in E\}$ to be the set of all edges incident on $v$.

Now set $E_1 = E$, and for each $i \in [2,k+1]$, set $E_i = \bigcup_{v \in V} Bunch_{i-1}(v)$.
Note that for $v\in V_i$ the random variable $|Bunch_i(v)|$ is dominated by a geometric random variable with parameter $p$, so $\E[|Bunch_i(v)|]\le {{s_i} \over {c\log n} }$, thus for any $v\in V$, w.h.p. $|Bunch_i(v)|\le O(s_i)$. We conclude that w.h.p.
$|E_i| = O(n \cdot s_{i-1})$.

\paragraph{The algorithm.}

The algorithm is to add to the spanner $H$ shortest path trees (SPT) from every vertex of $D_i$ in the graph $(V,E_i \cup E^*_i)$, and take all edges of $E_{k+1}$. See \algref{alg:+2W}.
\begin{algorithm}[ht]
	\caption{$\texttt{+2W spanner}(G,S,\epsilon)$}\label{alg:+2W}
	\begin{algorithmic}[1]
		\STATE Initialize $H \leftarrow \emptyset$;
		\FOR { $i = 1,2,\ldots,k$}
		\STATE Build SPT trees rooted at every vertex $v \in D_i$ in $(V,E_i \cup E^*_i)$, and add them to $H$;
		\ENDFOR
		\RETURN $H \cup E_{k+1}$;
	\end{algorithmic}
\end{algorithm}

We will also refer to each iteration $i$ of this for-loop as {\em step $i$ of the algorithm}. 

\subsection{Analysis of Size and Running Time}

For every index $i \in [k]$, we have w.h.p. $|D_i| = \tO(n/s_i)$, thus $\sum_{i=1}^k |D_i| \cdot n = \tO(n^2)\cdot \sum_{i=1}^k\frac{1}{s_i}=\tO(n^{3/2})$.
Also, w.h.p. $|E_{k+1}| \le n \cdot s_k = \tO(n^{3/2})$.
Hence the overall size of the spanner is $\tO(n^{3/2})$ as well.

To bound the running time, note that each step $i \in [k]$ of the algorithm requires computing $|D_i|$ SPTs in a graph with $O(|E_i|+n)$ edges. Using Dijkstra, each tree can be constructed in near linear time, so the total running time for step $i$ is
$$ \tO( |E_i| + n) \cdot|D_i| =\tO(n \cdot s_{i-1} \cdot n/s_i) =\tO(n^2)$$ time.
The last step requires $\tO(|E|)$ time, and thus the overall time is $\tO(n^2)$.

\subsection{Stretch Analysis}

Let $u,v$ be a vertex pair, let $P=P_{u,v}$ be the shortest $u-v$ path, and $W_{u,v}$ is the weight of the heaviest edge in $P$.
For the sake of the following lemma, step 0 of the algorithm is before the algorithm starts.
\begin{lemma}
	For every index $i = 0,1,\ldots,k$, at least one of the following holds:
	\begin{enumerate}
		\item $d_H(u,v) \le d_G(u,v) + 2 W_{u,v}$, or
		\item $E(P) \subseteq E_{i+1}$.
	\end{enumerate}
\end{lemma}
\begin{proof}
	The proof is by induction $i$.
	
	\noindent{\bf Base ($i= 0$):}
	Clearly $E(P) \subseteq E_1 = E$, i.e., the second assertion holds.
	
	\noindent{\bf Step:}
	Suppose that the induction hypothesis holds for some $i \in [0,k-1]$. 
	If the first assertion holds for $i$, then obviously the first assertion holds for $i+1$ as well. Hence, in this case we are done.
	
	So suppose that the second assertion holds for $i$, i.e., $E(P) \subseteq  E_{i+1}$.
	Consider the case that there exists an edge $e = (x,y) \in E(P) \setminus E_{i+2}$.
	(As otherwise $E(P) \subseteq E_{i+2}$, and the second assertion holds for $i+1$.)
	Then we claim that both $x,y \in V_{i+1}$. To see this, assume that, e.g., $x \not \in V_{i+1}$, but then by definition of Bunch for vertices not in $V_{i+1}$ we have that $(x,y) \in Bunch_{i+1}(x) \subseteq E_{i+2}$, contradiction.
	
	So we have $x,y \in V_{i+1}$, and $e = (x,y) \not \in Bunch_{i+1}(y)$.
	Thus $y' = p_{i+1}(y)$ is defined, and 
	$$W_{u,v} \ge w((x,y)) \ge w((y,y')) = w((y,p_{i+1}(y))~.$$
	Recall that $(y,p_{i+1}(y)) \in E^*_{i+1}$.
	So both paths $(y',y) \circ P(y,u)$ and $(y',y) \circ P(y,v)$ are contained in $E_{i+1} \cup E^*_{i+1}$.
	(We use $\circ$ here for concatenation,  $P(y,u)$ for the subpath of $P$ connecting $y$ with $u$, and
	$P(y,v)$ for the subpath of $P$ connecting $y$ with $v$.)
	
	Also, $y' \in D_{i+1}$. Hence inserting an SPT tree rooted at $y'$ in $E_{i+1} \cup E^*_{i+1}$ into the spanner $H$ guarantees 
	$$d_H(u,v) \le d_G(u,v) + 2w(y',y) \le d_G(u,v) + 2\cdot W_{u,v}~.$$
	This tree is indeed inserted into the spanner on step $i+1$, and so the first assertion for $i+1$ holds.
\end{proof}

Apply the lemma for $i=k$. If the first assertion holds, then we are done.
Otherwise $E(P) \subseteq E_{k+1}$. 
But then step $k+1$ of the algorithm ensures that $d_H(u,v) = d_G(u,v)$, as all edges of $E_{k+1}$ are inserted into $H$ on this step. This completes the proof of the following theorem.

\begin{theorem}
	Let $G=(V,E,w)$ be a weighted graph with $n$ vertices, then there is an $\tO(n^2)$ time randomized algorithms that produces w.h.p. a $+2W$ spanner of size $\tO(n^{3/2})$.
\end{theorem}

\section{A $+(2 + \varepsilon)W$ subsetwise spanner}\label{sec:subset}

Let $G=(V,E,w)$ be a weighted undirected graph, a parameter $0<\varepsilon<1$, and $S\subseteq V$ a set of vertices. In this section we devise an $+(2 + \varepsilon)W$ subsetwise spanner of size $O(n\cdot\sqrt{|S|}/\varepsilon)$. That is, the spanner guarantees an additive stretch at most $(2+\varepsilon)\cdot W_{u,v}$ for any $u,v\in S$.

\paragraph*{Construction.}
Our algorithm follows a similar greedy idea to our previous constructions.
We start by letting $H$ be a $( \sqrt{|S|} )$-light initialization of $G$. Next, sort all the pairs $\{u, v\} \in {S\choose 2}$ by $W_{u, v}$ in increasing order, breaking ties arbitrarily. When considering $u, v$, we add $P_{u, v}$ to $H$ if 
\begin{align} \label{eq:subsetwise}
	d_H(u, v) > d_G(u, v) + (2 + \varepsilon) W_{u, v}.
\end{align}

\paragraph*{Analysis. }
Our main lemma is a variant of \lemmaref{lem:improved_dist} tailored to the subsetwise case.  For every path added to $H$, we improve the distance from many neighbors of the path to vertices in $S$, and have a good guarantee for all of them. Note that even though we claim improvements for many pairs in $S\times V$, the final spanner {\em does not} have guarantee for all such pairs, only to those in $S\times S$.

\begin{lemma}\label{lem:improved_subsetwise}
	Let $P_{u, v}$ be a path that was added to $H$. Denote by $H_0$ the spanner just before $P_{u, v}$ was added and by $H_1$ the spanner right after the path was added. Let $a$ be a $(\sqrt{|S|})$-light neighbor of $x\in P_{u,v}$ with $w(a, x) \le W_{u, v}$. Then both of the following holds.
	\begin{enumerate}
		
		\item $
		d_{H_1}(u, a) \le d_G(u, a) + 2 W_{u, v}
		\textbf{ and }
		d_{H_1}(v, a) \le d_G(u,a) + 2 W_{u, v}$.
		\item
		$
		d_{H_1}(u, a) \le d_{H_0}(u, a) - \frac{\varepsilon}{2} W_{u, v}
		\textbf{ or }
		d_{H_1}(v,a) \le d_{H_0}(v, a) - \frac{\varepsilon}{2} W_{u, v}
		$.

	\end{enumerate}
\end{lemma}
\begin{proof}
	We begin with the first item. By the triangle inequality,
	\begin{align*}
		d_{H_1}(u, a) &\le d_{H_1}(u, x) + d_{H_1}(x, a) \\
		&= d_{G}(u, x) + d_{G}(x, a) \\
		&\le d_{G}(u, a) + d_G(x, a) + d_{G}(x, a) \\
		&\le d_{G}(u, a) + 2W_{u, v}. 
	\end{align*}

	The bound on $d_{H_1}(v,a)$ follows in a symmetric manner, which concludes the proof of the first item.

	Seeking contradiction, assume that the second item does not hold. This suggests that
	\[
	d_{H_0}(u, a)< d_{H_1}(u, a)+ \frac{\varepsilon}{2} W_{u, v} \le d_G(u, x) + (1 + \frac{\varepsilon}{2})W_{u, v}~,
	\]

	and also
	\[
	d_{H_0}(v,a) <  d_{H_1}(v,a)+ \frac{\varepsilon}{2} W_{u, v}\le d_G(v, x) + (1 + \frac{\varepsilon}{2})W_{u, v}~.
	\]
	 So we have that
	\begin{align*}
		d_{H_0}(u, v)
		&\le d_{H_0}(u, a) + d_{H_0}(a, v) \\
		&< d_{G}(u, x) + (1 + \frac{\varepsilon}{2})W_{u, v}  + d_{G}(x, v) +(1 + \frac{\varepsilon}{2})W_{u, v} \\
		&= d_G(u, v) + (2 + \varepsilon) W_{u, v},
	\end{align*}
which is a contradiction to \eqref{eq:subsetwise}, since we assumed that the path $P_{u, v}$ was added to the spanner.
\end{proof}
	
\begin{theorem}
	For every undirected weighted graph $G= (V, E, w)$ with $n$ vertices, a vertex set $S \subseteq V$ and a parameter $\varepsilon>0$, there exists a deterministic polynomial time algorithm that produces a $+(2 + \varepsilon)W$ subsetwise $S \times S$ spanner of size $O(\frac{1}{\varepsilon}\cdot n \sqrt{|S|})$.
\end{theorem}

\begin{proof}
Our algorithm clearly yields a $+(2+\varepsilon)\cdot W$ spanner for $S\times S$, and can be done in polynomial time. It remains to bound the size of the spanner. The $(\sqrt{|S|})$-initialization adds at most $n\cdot\sqrt{|S|}$ edges to $H$.

Let $u,v\in S$ be such that $P_{u,v}$ is added to the spanner. Let $H_0$ be the spanner just before the path is added, and $H_1$ after. A pair $(a,b)$ in $S\times V$ is said to {\em set-off} if this is the first time that $d_{H_1}(a,b)\le d_G(a,b)+2W_{u,v}$. This pair is {\em improved} if $d_{H_1}(a,b)\le d_{H_0}(a,b)-\frac\varepsilon2\cdot W_{u,v}$.

By \lemmaref{lem:a_lot_of_neighbors} if there are $\ell$ missing edges of $P_{u,v}$ in $H_0$, then there are at least $\Omega(\ell\cdot\sqrt{|S|})$ light neighbors that are connected to vertices on missing edges of $P_{u,v}$ with weight at most $W_{u,v}$. Thus there are $\Omega(\ell\cdot\sqrt{|S|})$ choices for $a$ in \lemmaref{lem:improved_subsetwise}. That is, so many pairs in $S\times V$ are set-off and improved.
We notice that pairs from $S \times V$ can be set-off once and improved at most $\frac{4}{\varepsilon}$ times thereafter. If $t$ is the total number of edges added to $H$ by all the paths in the second stage of the algorithm, we get that
\[
\Omega(t \cdot \sqrt{|S|}) \le O(\frac{|S| \cdot |V|}{\varepsilon})~,
\]
thus $t = O(\frac{1}{\varepsilon} \cdot n\sqrt{|S|})$.
\end{proof}




\section{A $+4W$ emulator}\label{sec:4emul}

\paragraph{Construction}
Our algorithm for a $+4W$ emulator works as follows. Start by letting $H=(V,E',d_G)$ be a $(2n^{1/3} \ln{n})$-light initialization of $G$.\footnote{By increasing the leading constant from 2 to $c$, we can reduce the failure probability to at most $O(n^{1-c})$.}
Let $S\subseteq V$ be a random set, created by sampling each vertex independently with probability $\frac{1}{n^{1/3}}$.
We finish by adding $S \times S$ to $E'$ (with weights corresponding to distances in $G$).


\begin{theorem}
	For every undirected weighted graph $G= (V, E, w)$, there exists a randomized algorithm that produces w.h.p. a $+4W$ emulator of size $O(n^{4/3} \log n)$.
\end{theorem}
\begin{proof}
	We begin with the stretch analysis. Let $u, v \in V$. If all the edges of $P_{u, v}$ exists in $H$, then $d_H(u, v) = d_G(u, v)$ and we are done.
	
	Otherwise, let $u=x_1, x_2, \dots x_k=v$ be the vertices of $P_{u, v}$ sorted by their distance from $u$.
	Let $x_i, x_j$ be the first and last vertices for which $\{x_i, x_{i+1}\}, \{x_{j-1}, x_{j}\} \notin E'$.

We claim that each of $x_i, x_j$ have at least $2n^{1/3} \ln n$ neighbors in $G$, because $\{x_i, x_{i+1}\}, \{x_{j-1}, x_{j}\}$ were not included in $H$ as part of the light initialization.
	By \lemmaref{lem:emu:dominating}, there exists $a, b \in S$ which are $(2n^{1/3} \ln n)$-light neighbors of $x_i, x_j$ respectively.
	In addition, $x_{i+1} ,x_{j-1}$ are not $(2n^{1/3} \ln n)$-light neighbors of $x_{i}, x_j$, respectively, thus $w(x_i, a) \le w(x_i, x_{i+1}) \le W_{u,v}$ and $w(x_j, b) \le w(x_{j-1}, x_j) \le W_{u, v}$.

The sub-paths $P_{u, x_i}, P_{x_j, v}$ exist in $H$, and also all the edges $\{x_i, a\}, \{a, b\}, \{b, x_j\} \in E'$. We can use them for bounding $d_H(u, v)$ (see figure \ref{fig:theo:4W_emu}).

	\begin{align*}
		d_H(u, v)
		&\le d_H(u, x_i) + d_H(x_i, a) + d_H(a, b) + d_H(b, x_j) + d_H(x_j, v) \\
		&= d_G(u, x_i) + d_G(x_i, a) + d_G(a, b) + d_G(b, x_j) + d_G(x_j, v) \\
		&\le d_G(u, x_i) + d_G(x_i, a) + d_G(x_i, a) + d_G(x_i, x_j) + d_G(b, x_j) + d_G(b, x_j) + d_G(x_j, v)\\
		&\le d_G(u, v) + 4 W_{u, v}.
	\end{align*}

	Bounding the size is straightforward. The $n^{1/3} \log n$-light initialization introduces at most $O(n^{4/3} \log n)$ edges, while $|S|$ is a Bernoulli random variable with parameters $(n, \frac{1}{n^{1/3}})$.
	Therefore, $E[|S|]=n \cdot \frac{1}{n^{1/3}} = n^{2/3}$ and by Chernoff bound $|S| \le 2n^{2/3}$, w.h.p..
	Thus $|S \times S| = O(n^{2/3} \cdot n^{2/3}) = O(n^{4/3})$ w.h.p..
	
	Hence the total size of the emulator is $O(n^{4/3} \log n)$ w.h.p..
\end{proof}
\begin{figure}[h]
	\centering
	\includegraphics{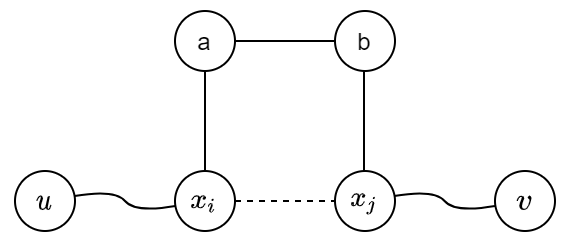}	
	\caption{Straight lines are edges available in $H$. Curved lines are shortest paths available in $H$}
	\label{fig:theo:4W_emu}
\end{figure}

\clearpage

\bibliographystyle{alpha}
\bibliography{12spanner}
\end{document}